\documentclass[runningheads,a4paper]{llncs}

\usepackage{amssymb}
\usepackage{graphicx}
\usepackage[linesnumbered,ruled,vlined,boxed]{algorithm2e}
\usepackage{amsmath,amssymb,amstext}
\usepackage{caption}
\usepackage{subcaption}
\usepackage[caption=false]{subfig}

\usepackage{comment}
\usepackage{todonotes}

\usepackage[margin=1in]{geometry}

\usepackage{cite}

\newcommand{\ie}{\emph{i.e.}\xspace}
\newcommand{\eg}{\emph{e.g.}\xspace}

\usepackage{url}
\urldef{\mailsa}\path|{alfred.hofmann, ursula.barth, ingrid.haas, frank.holzwarth,|
\urldef{\mailsb}\path|anna.kramer, leonie.kunz, christine.reiss, nicole.sator,|
\urldef{\mailsc}\path|erika.siebert-cole, peter.strasser, lncs}@springer.com|

\begin{document}
\frenchspacing
\mainmatter  

\title{Strategyproof Quota Mechanisms for Multiple Assignment Problems}

\titlerunning{Strategyproof Quota Mechanisms for Multiple Assignment Problems}

%
%

\author{Hadi Hosseini\inst{1} \and Kate Larson\inst{2}}
\authorrunning{Hadi Hosseini \and Kate Larson}


\institute{Computer Science Department,
Carnegie Mellon University\\
\email{hhosseini@cmu.edu}	
\and	
	Cheriton School of Computer Science,
	University of Waterloo\\
	\email{klarson@uwaterloo.ca}
}

%

%
%

\toctitle{Lecture Notes in Computer Science}
\tocauthor{Authors' Instructions}
\maketitle

\begin{abstract}
We study the problem of allocating multiple objects to agents without transferable utilities, where each agent may receive more than one object according to a quota. 
Under lexicographic preferences, we characterize the set of strategyproof, non-bossy, and neutral quota mechanisms and show that under a mild Pareto efficiency condition, serial dictatorship quota mechanisms are the only mechanisms satisfying these properties. Dropping the neutrality requirement, this class of quota mechanisms further expands to sequential dictatorship quota mechanisms.
We then extend quota mechanisms to randomized settings, and show that the random serial dictatorship quota mechanisms (RSDQ) are envyfree, strategyproof, and ex post efficient for any number of agents and objects and any quota system, proving that the well-studied Random Serial Dictatorship (RSD) satisfies envyfreeness when preferences are lexicographic.
 
\end{abstract}

\section{Introduction}

We consider the problem of allocating indivisible objects to agents without any explicit market.
In many real-life domains such as course assignment, room assignment, school choice, medical resource allocation, etc. the use of monetary transfers or explicit markets are forbidden because of ethical and legal issues.
Much of the literature in this domain is concerned with designing incentive compatible mechanisms that incentivizes agents to reveal their preferences truthfully. Moreover, the criterion of Pareto efficiency along with strategyproofness provide stable solutions to such allocation problems.

We are interested in allocation problems where each agent may receive a set of objects and thus we search for mechanisms that satisfy some core axiomatic properties of strategyproofness, Pareto efficiency, and non-bossiness. 
Examples of such allocation problems include distributing inheritance among heirs\footnote{Here we only consider non-liquid assets that cannot be quickly or easily converted to transferable assets such as money.}, allocating multiple tasks to employees, assigning scientific equipment to researchers, assigning teaching assistants to different courses, and allocating players to sports teams. The common solution for allocating players to teams or allocating courses to students in the course assignment problem is the Draft mechanism~\cite{brams1979prisoners}, where agents choose one item in each picking round. However, allocation mechanisms, such as the Draft mechanism, have been shown to be highly manipulable in practice and fail to guarantee Pareto optimality~\cite{budish2012multi}.

Svensson~\cite{svensson1994queue,svensson1999strategy} formulated the \emph{standard assignment problem} (first proposed by Shapley and Scarf~\cite{shapley1974cores}) where each agent receives exactly one item, and showed that Serial Dictatorship mechanisms are the only social choice rules that satisfy Pareto efficiency, strategyproofness, non-bossiness, and neutrality. 
In contrast to the standard assignment problem, in the \emph{multiple assignment problem} agents may require bundles or sets of objects according to a predefined quota and might have various interesting preferences (\eg complements or substitutes) over these sets. However, the class of sequential dictatorships mechanisms no longer characterizes all non-bossy, Pareto efficient, and strategyproof social choice mechanisms.
To address this issue, P{\'a}pai~\cite{papai2000strategyproof} and Hatfield~\cite{hatfield2009strategy} studied the multiple assignment problem where objects are assigned to agents subject to a quota. P{\'a}pai~\cite{papai2000strategyproof} showed that under quantity-monotonic preferences every strategyproof, non-bossy, and Pareto efficient social choice mechanism is sequential; while generalizing to monotonic preferences, the class of such social choice functions gets restricted to \emph{quasi-dictatorial} mechanisms where every agent except the first dictator is limited to pick at most one object. P{\'a}pai's characterization is essentially a negative result and rules out the possibility of designing neutral, non-bossy, strategyproof, and Pareto efficient mechanisms that are not strongly dictatorial.
Hatfield \cite{hatfield2009strategy}, on the other hand, addressed this issue by assuming that all agents have precisely fixed and equal quotas, and showed that serial dictatorship is strategyproof, Pareto efficient, non-bossy, and neutral for responsive preferences.

Our work generalizes these results, for a subclass of preferences, by allowing any number of agents or objects, and assuming that individual agents' quotas can vary and be agent specific, imposing no restrictions on the problem size nor quota structures.
Instead, we are interested in expanding the possible quota mechanisms to a larger class, essentially enabling a social planner to choose any type of quota system based on a desired metric such as seniority. Our main focus is on the \emph{lexicographic preference} domain, where agents have idiosyncratic private preferences.

Lexicographic preferences \cite{fishburn1975axioms} have recently attracted attention among researchers in economics and computer science~\cite{schulman2012allocation,saban2014note,kohli2007representation}. 
In behavioral economics and psychology as well as consumer market research, there is evidence for the presence of lexicographic preferences among individuals such as breaking ties among equally valued alternatives~\cite{drolet2004rationalizing}, making purchasing decisions by consumers~\cite{colman1999singleton}, and examining public policies, job candidates, etc.~\cite{tversky1988contingent}. Choice and decisions are particularly tend to look more lexicographic in ordinal domains, thus, in ordinal mechanism design one must pay particular attention to the settings wherein agents may treat alternatives as non-substitutable goods.
%
%

%
Our main results in the lexicographic preference domain are the following:

\begin{itemize}
	\item We characterize the set of strategyproof, non-bossy, and neutral allocation mechanisms when there is a quota system. We show that serial dictatorships are the only mechanisms satisfying our required properties of strategyproofness, non-bossiness, Pareto efficiency, and neutrality. Allowing any quota system enables the social planner to remedy the inherent unfairness in deterministic allocation mechanisms by assigning quotas according to some fairness criteria (such as seniority, priority, etc.).
	
	\item We generalize our findings to \emph{randomized mechanisms} and show that random serial dictatorship quota mechanisms (RSDQ) satisfy strategyproofness, ex post efficiency, and envyfreeness in the domain of lexicographic preferences. Consequently, we prove that the well-known Random Serial Dictatorship (RSD) mechanism in standard assignment settings satisfies envyfreeness when preferences are lexicographic. 
	Thus, random quota mechanisms provide a rich and extended class for object allocation with no restriction on the market size nor quota structure while providing envyfreeness in the lexicographic domains, justifying the use of such mechanisms in many practical applications.
	%
	%
	
\end{itemize}

\section{Related Work}


In the \emph{standard assignment} problem (sometimes known as the house allocation problem), each agent is entitled to receive exactly one object from the market. P{\'a}pai~\cite{papai2000hierarchical} extended the standard model of Svensson~\cite{svensson1994queue,svensson1999strategy} to settings where there are potentially more objects than agents (each agent receiving at most one object) with a hierarchy of endowments, generalizing Gale's top trading cycle procedure. This result showed that the hierarchical exchange rules characterize the set of all Pareto efficient, group-strategyproof, and reallocation proof mechanisms.
%
%
In the \emph{multiple-assignment} problem, agents may receive sets of objects, and thus, might have various interesting preferences over the bundles of objects. P{\'a}pai~\cite{papai2001strategyproof} studied this problem on the domain of strict preferences allowing for complements and substitutes, and showed that sequential dictatorships are the only strategyproof, Pareto optimal, and non-bossy mechanisms. Ehlers and Klaus~\cite{ehlers2003coalitional} restricted attention to responsive and separable preferences and essentially proved that the same result persists even in a more restrictive setting. Furthermore, Ehlers and Klaus showed that considering resource monotonic allocation rules, where changing the available resources (objects) affects all agents similarly, limits the allocation mechanisms to serial dictatorships.
Hylland and Zeckhauser's pseudo-market design based on eliciting cardinal utilities~\cite{hylland1979efficient} and its deterministic counterpart based on competitive equilibrium from equal incomes (CEEI)
 provide efficient and envyfree solutions but are highly susceptible to manipulation. Zhou~\cite{zhou1990conjecture}, based on Gale's conjecture~\cite{gale1987college}, proved that there do not exist (randomized) allocation rules that satisfy symmetry, Pareto efficiency, and strategyproofness.

In the randomized settings, Random Serial Dictatorship (RSD) and Probabilistic Serial Rule (PS) are well-known for their prominent economic properties. RSD satisfies strategyproofness, ex post efficiency, and equal treatment of equals \cite{abdulkadirouglu1998random}, while PS is ordinally efficient and envyfree but not strategyproof~\cite{bogomolnaia2001new}.
For divisible objects, Schulman and Vazirani~\cite{schulman2012allocation} showed that if agents have lexicographic preferences, the Probabilistic Serial rule is strategyproof under strict conditions over the minimum available quantity of objects and the maximum demand request of agents. Under indivisible objects, these strict requirements translate to situations where the number of agents is greater than the number of objects and each agent receives at most one object. When allocating multiple objects to agents, Kojima~\cite{kojima2009random} obtained negative results on (weak) strategyproofness of PS in the general domain of preferences. Not only PS is not strategyproof, but the fraction of manipulable profiles quickly goes to one as the number of objects exceeds that of agents, even under lexicographic preferences~\cite{hosseini2016investigating}.
%
%
In contrast, we seek to find strategyproof and envyfree mechanisms with no restriction on the number of agents or objects under the lexicographic preference domain, addressing the open questions in ~\cite{papai2000strategyproof} and in~\cite{schulman2012allocation} about the existence of a mechanism with more favorable fairness and strategyproofness properties.

\section{The Model}

There is a set of $m$ indivisible objects $M = \{1,\ldots, m \}$ and a set of $n$ agents $N = \{1,\ldots, n \}$.
There is only one copy of each object available, and an agent may receive more than one object. 
Let $\mathbb{M} = \mathbb{P}(M)$ denote the power set of $M$.

Agents have private preferences over sets of objects. 
Let $\mathcal{P}$ denote the set of all complete and strict preferences over $\mathbb{M}$.
Each agent's preference is assumed to be a strict relation $\succ_{i}\in \mathcal{P}$. 
A \emph{preference profile} denotes a preference ordering for each agent and is written as $\succ = (\succ_{1},\ldots, \succ_{n}) \in \mathcal{P}^{n}$. Following the convention, $\succ_{-i} = (\succ_{1},\ldots, \succ_{i-1},\succ_{i+1},\ldots, \succ_{n}) \in \mathcal{P}^{n}$, and thus $\succ = (\succ_{i}, \succ_{-i})$.

An allocation is a $n \times m$ matrix $A \in \mathcal{A}$ that specifies a (possibly probabilistic) allocation of objects to agents. 
The vector $A_{i} = (A_{i,1}, \ldots, A_{i,m})$ denotes the allocation of agent $i$, that is, 
\begin{small}
	\begin{equation*}
		A = 
		\begin{pmatrix}
			A_{1} \\
			A_{2} \\
			\vdots \\
			A_{n} 
		\end{pmatrix}
		=
		\begin{pmatrix}
			A_{1,1} & A_{1,2} & \ldots & A_{1,m} \\
			A_{2,1} & A_{2,2} & \ldots & A_{2,m} \\
			\vdots & \vdots & \ddots & \vdots \\
			A_{n,1} & A_{n,2} & \ldots & A_{n,m} 
		\end{pmatrix}
	\end{equation*}
\end{small}
We sometimes abuse the notation and use $A_{i}$ to refer to the set of objects allocated to agent $i$.
Let $\mathcal{A}$ refer to the set of possible allocations. Allocation $A \in \mathcal{A}$ is said to be \textit{feasible} if and only if $\forall j\in M, \sum_{i\in N} A_{i,j} = \{0, 1\}$, no single object is assigned to more than one agent, while some objects may not be assigned. Note that we allow \emph{free disposal}, and therefore, $\bigcup_{i\in N} A_{i} \subseteq M$.
For two allocations we write $A_i \succ_i B_i$ if agent $i$ with preferences $\succ_{i}$ strictly prefers $A_{i}$ to $B_{i}$. Thus, $A_i \succeq_i B_i$ and $B_i \succeq_i A_i$ implies $A_{i} = B_{i}$.

Preference $\succ_{i}$ is \emph{lexicographic} if there exists an ordering of objects, $(a, b ,c , \ldots)$, such that for all $A, B \in \mathcal{A}$ if $a \in A_{i}$ and $a \notin B_{i}$ then $A_{i} \succ_{i} B_{i}$; if $b \in A_{i}$ and $a, b \notin B_{i}$ then $A_{i} \succ_{i} B_{i}$; and so on.
That is, the ranking of objects determines the ordering of the sets of objects in a lexicographic manner. Note that lexicographic preferences are \emph{responsive} and \emph{strongly monotonic}. A preference relation is \emph{responsive} if $A_{i} \bigcup B_{i} \succ_{i} A_{i} \bigcup B'_{i}$ if and only if $B_{i} \succ_{i} B'_{i}$. Strong monotonicity means that any set of objects is strictly preferred to all of its proper subsets. 
We make no further assumption over preference relations.

An allocation mechanism is a function $\pi: \mathcal{P}^{n} \to \mathcal{A}$, which assigns a feasible allocation to every preference profile. Thus, agent $i$'s allocation $A_{i}$ can also be represented as $\pi_{i}$.
An allocation mechanism assigns objects to agents according to a \emph{quota system} $q$, where $q_{i}$ is the quota of the $i$th dictator such that $\sum_{i= 1}^{n} q_{i} \leq m$. Since in our model not all agents need to be assigned an object, we use the size of quota $|q|$ to denote the number of agents that are assigned at least one object. Thus, we have $|q| \leq n$. 
From the revelation principle~\cite{dasgupta1979implementation}, we can restrict our analysis to direct mechanisms that ask agents to report their preferences to the mechanism directly.

\subsection{Properties}

In the context of deterministic assignments, an allocation $A$ \emph{Pareto dominates} another allocation $B$ at $\succ$ if $\exists i\in N$ such that $A_{i} \succ_{i} B_{i}$ and $\forall j\in N$ $A_{j} \succeq_{j} B_{j}$. An allocation is \emph{Pareto efficient} at $\succ$ if no other allocation exists that Pareto dominates it at $\succ$.
Since a social planner may decide to only assign $C \leq m$ number of objects, we need to slightly modify our efficiency definition.
We say that an allocation that assigns $C = \sum_{i=1}^{n} q_{i}$ objects is Pareto C-efficient if there exists no other allocation that assigns an equal number of objects, $C$, that makes at least one agent strictly better off without making any other agent worse off. A Pareto C-efficient allocation is also  Pareto efficient when $\sum_{i=1}^{n} q_{i} = m$.

\begin{definition} [Pareto C-efficiency]
	A mechanism $\pi$ with quota $q$, where $C=\sum_{i} q_{i}$, is Pareto C-efficient if for all $\succ\in \mathcal{P}^{n}$, there does not exist $A \in \mathcal{A}$ which assigns $C$ objects such that for all $i\in N$, $A_{i} \succeq_{i} \pi_{i}(\succ)$, and $A_{j} \succ_{j} \pi_{j}(\succ)$ for some $j\in N$.
\end{definition}

A mechanism is strategyproof if there exists no non-truthful preference ordering $\succ'_{i} \neq \succ_{i}$ that improves agent $i$'s  allocation. More formally, 
\begin{definition} [Strategyproofness]
	Mechanism $\pi$ is strategyproof if for all $\succ \in \mathcal{P}^{n}$, $i \in N$, and for any misreport $\succ'_{i} \in \mathcal{P}$, we have $\pi_{i}(\succ) \succeq_{i} \pi_{i}(\succ'_{i},\succ_{-i})$.
\end{definition}

Although strategyproofness ensures that no agent can benefit from misreporting preferences, it does not prevent an agent from reporting a preference that changes the prescribed allocation for some other agents while keeping her allocation unchanged. This property was first proposed by Satterthwaite and Sonnenschein~\cite{satterthwaite1981strategy}.
A mechanism is \emph{non-bossy} if an agent cannot change the allocation without changing the allocation for herself.
\begin{definition} [Non-bossiness]
	A mechanism is non-bossy if for all $\succ\in \mathcal{P}^{n}$ and agent $i\in N$, for all $\succ'_{i}$ such that $\pi_{i}(\succ) = \pi_{i}(\succ'_{i},\succ_{-i})$ we have $\pi(\succ) = \pi(\succ'_{i},\succ_{-i})$.
\end{definition}

Non-bossiness and strategyproofness only prevent certain types of manipulation; changing another agent's allocation or individually benefiting from a strategic report. However, it may still be possible for two or more agents to form a coalition and affect the final outcome, so that at least one of them improves her allocation ex post. This requirement is called group-strategyproofness, which precludes group manipulation as well as individual agent manipulation. 


Our last requirement is neutrality. 
Let $\phi: M \to M$ be a permutation of the objects.
For all $A\in\mathcal{A}$, let $\phi(A)$ be the set of objects in $A$ renamed according to $\phi$. Thus, $\phi(A) = (\phi(A_{1}), \ldots, \phi(A_{n}))$.
For each $\succ \in \mathcal{P}^{n}$ we also define $\phi(\succ) = (\phi(\succ_{1}), \ldots, \phi(\succ_{n}))$ as the preference profile where all objects are renamed according to $\phi$. 

\begin{definition} [Neutrality]
	A mechanism $\pi$ is neutral if for any permutation function $\phi$ and for all preference profiles $\succ\in\mathcal{P}^{n}$, $\phi(\pi(\succ)) = \pi(\phi(\succ))$.
\end{definition}

In other words, a mechanism is \emph{neutral} if it does not depend on the name of the objects, that is, changing the name of some objects results in a one-to-one identical change in the outcome. 
It is clear that above conditions reduce the set of possible mechanisms drastically.

\section{Allocation Mechanisms}

Several plausible multiple allocation mechanisms exploit interleaving picking orders to incorporate some level of fairness, where agents can take turns each time picking one or more objects  \cite{kohler1971class,brams2000win,brams2005efficient}. An interleaving mechanism alternates between agents, allowing a single agent to pick objects in various turns.
The \emph{interleaving mechanisms} have been widely used in many everyday life activities such as assigning students to courses, members to teams, and in allocating resources or moving turns in boardgames or sport games. To name a few, \emph{strict alternation} where agents pick objects in alternation (\eg 1212 and 123123) and \emph{balanced alternation} where the picking orders are mirrored (\eg agent orderings 1221 and 123321), and the well-known \emph{Draft mechanism}~\cite{budish2012multi,brams1979prisoners,brams2000win} that randomly chooses a priority ordering over $n$ agents and then alternates over the drawn priority ordering and its reverse sequence are the examples of such mechanisms.
However, all these interleaving mechanisms are highly manipulable in theory; computing optimal manipulations under interleaving mechanisms is shown to be easy only for two agents under additive and separable preferences and similarly for lexicographic preferences~\cite{DBLP:journals/corr/AzizBLM16}.
Extending to non-separable preferences, deciding a strategic picking strategy is NP-complete, even for two agents~\cite{bouveret2014manipulating}. Kalinowski et al.~\cite{kalinowski2013social} studied interleaving mechanisms (alternating policies) from a game-theoretical perspective and showed that under linear order preferences the underlying equilibrium in a two-person picking game is incentive compatible~\cite{kalinowski2013strategic}. Nonetheless, such interleaving mechanisms have been shown to be heavily manipulated in practice~\cite{budish2012multi}.

Before discussing the main characterization results, here we provide a formal statement for all interleaving mechanisms.
We generalize such allocation procedures to any mechanism with an interleaving order of agents with general preferences where at least one agent gets to choose twice, once before and once after one (or more) agents.  We note that all missing proofs can be found in the appendix.

\begin{theorem}\label{thm:noInterleaving}
	There exists no interleaving mechanism that satisfies Pareto C-efficiency, non-bossiness, and strategyproofness.
\end{theorem}

Clearly, an \emph{imposed mechanism} that assigns a fixed allocation to every preference profile is strategyproof and non-bossy but does not satisfy Pareto C-efficiency~\cite{papai2001strategyproof}.\footnote{An imposed mechanism does not take agents' preferences into account and prescribes the same allocation to every preference profile.}
With these essentially negative results for interleaving mechanisms, we restrict our attention to the class of sequential dictatorship mechanisms, where each agent only gets one chance to pick (possibly more than one) objects. 

\subsection{Sequential Mechanisms}

Let $q$ denote a quota system such that $\sum_{i} q_{i} \leq m$.
In a sequential dictatorship mechanism with quota $q$, the first dictator chooses $q_{1}$ of her most preferred objects; the second dictator is chosen  depending on the set of objects allocated to the first dictator. The second dictator then chooses $q_{2}$ objects of her most preferred objects among the remaining objects. This procedure continues, where the choice of the next dictator may be determined based on the earlier allocations, until no object or no agent is left.


Let $f$ be a function that, given a partial allocation of objects to some agents, returns the next dictator. Then, $f_{i}(\cdot) = j$ means that agent $j$ is ranked $i$th in the ordering of dictators. There exists an agent $f_{1}$ (first dictator) for each preference profile $\succ\in\mathcal{M}$, and an ordering of the remaining dictators such that the $i$th dictator is identified recursively by 
\begin{equation*}
	f_{i}(\pi_{f_{1}}(\succ), \ldots, \pi_{f_{i-1}}(\succ))
\end{equation*}
In other words, the choice of the next dictator only depends on the previous dictators and their allocation sets and \textit{does not} depend on the preferences of the previous dictators.
The following example shows why the choice of dictator should not depend on previous dictators' preferences.
\begin{example}
	Assume three agents and four objects with $q=(2,1,1)$ and consider the following rule for identifying the order of the dictators: if the first dictator's preference is $a \succ b \succ c \succ d$ then the ordering of other agents is (2,3), otherwise the order is (3,2). Now if agent 2 and 3 have identical preferences as agent 1, then agent 1 can simply change agent 2 and 3's allocations by misrepresenting her preference as $\hat{\succ}_{1}: b \succ a \succ c \succ d$ while her allocation remains unchanged. Thus, this sequential dictatorship mechanism is bossy even though it satisfies Pareto efficiency and strategyproofness.
\end{example}

\begin{definition} [Sequential Dictatorship]
	Let $\mathbb{M}_{k} = \mathbb{P}_{\leq k}(M)$ be the set of subsets of $M$ of cardinality less than or equal $k$. 
	An allocation mechanism $\pi: \mathcal{P}^{n} \to \mathcal{A}$ is a sequential dictatorship quota mechanism if there exists a quota system $q$ and an ordering $f$ such that for all $\succ\in \mathcal{P}^{n}$, 
	\begin{align*}
		\pi_{f_{1}}(\succ) = & \{ Z \in \mathbb{M}_{q_{1}} | Z \succ_{1} Z' ~\mbox{for all}~ Z'\in \mathbb{M}_{q_{1}} \}\\
		\pi_{f_{i}(\pi_{f_{1}},\ldots, \pi_{f_{i-1}})} (\succ) = & \{ Z \in \mathbb{M}_{q_{i}} \setminus \bigcup_{j=1}^{j=i-1} \pi_{f_{j}}(\succ) | Z \succ_{f_{i}} Z' ~\mbox{for all}~ |Z'| = |q_{i}| \}
	\end{align*}
\end{definition}
A \textbf{serial dictatorship mechanism }is an example of a sequential mechanism where the ordering is a permutation of the agents, determined a priori, that is, 
for all $\succ\in \mathcal{P}^{n}$, $\pi_{f(\cdot)}(\succ) =  \pi_{f}(\succ)$. Such mechanisms satisfy neutrality.
From now on, we simply use the vector $f$ instead of $f(\cdot)$ when the ordering is predefined independent of the choice of objects.





\section{Serial Dictatorship Quota Mechanisms}

In this section, we study serial dictatorship mechanisms for quota allocations and characterize the set of strategyproof, non-bossy, neutral, and Pareto efficient mechanisms subject to various quota systems.

When allocating objects sequentially via a quota system $q$, Pareto C-efficiency requires that no two agents be envious of each others' allocations since then they can simply exchange objects ex post, implying that the initial allocation is dominated by the new allocation after the exchange. For example, take a serial dictatorship with $q_{1} = 1$ and $q_{2} = 2$ and three objects. Agent $1$ will receive her top choice object $\{a\}$ (since $\{a\} \succ_{1} \{b\} \succ_{1} \{c\}$) according to her preference and agent 2 receives $\{b,c\}$. However, it may be the case that $\{b,c\} \succ_{1} \{ a \}$ while $\{ a \} \succ_{2} \{b,c\}$ and both agents may be better off exchanging their allocations. Thus, we have the following proposition for general preferences.

\begin{proposition}\label{prop:pareto}
	For general preferences, sequential (and serial) dictatorship quota mechanisms do not guarantee Pareto C-efficiency.
\end{proposition}

In the absence of Pareto C-efficiency in the domain of general preferences, a social planner is restricted to use only one type of quota system; either assigning at most one object to all agents except the first dictator (who receives the remaining objects), or setting equal quotas for all agents~\cite{papai2000strategyproof,hatfield2009strategy}.

Due to the impossibility shown in Proposition~\ref{prop:pareto},  we restrict ourselves to the interesting class of lexicographic preferences. We show that if preferences are lexicographic, regardless of the selected quota system, any serial dictatorship mechanism guarantees Pareto C-efficiency. We first provide the following lemma in the lexicographic domain.

\begin{lemma}\label{Lem:lex}
	The following statements hold for two sets of objects when preferences are lexicographic:
	\begin{itemize}
		\item[-] If $B_{i} \subset A_{i}$ then $A_{i} \succ_{i} B_{i}$.
		\item[-] For all $X$ such that $X \cap A_{i} = \emptyset$, we have $A_{i} \succ_{i} B_{i}$ iff $A_{i} \cup X \succ_{i} B_{i} \cup X$.
		
		\item[-] If $B_{i} \not\subset A_{i}$ and $A_{i} \succ_{i} B_{i}$ then there exists an object $x\in A_{i}$ such that $x \succ_{i} X$ for all $X \in  \mathbb{P}(B_{i} - A_{i})$.
	\end{itemize}
\end{lemma}

\begin{proposition}\label{prop:ParetoEfficient}
	If preferences are lexicographic, the serial dictatorship quota mechanism is Pareto C-efficient.
\end{proposition}
\begin{proof}
	Consider a mechanism $\pi$ with quota $q$, that assigns $C = \sum_{i} q_{i}$ objects.
	Suppose for contradiction that there exists an allocation $B$ with arbitrary quota $q'$, where $C' = \sum_{i} q'_{i}$, that Pareto dominates $A = \pi(\succ)$. We assume $C' = C$ to ensure that both allocations assign equal number of objects (Otherwise by strong monotonicity of lexicographic preferences and Lemma \ref{Lem:lex} one can assign more objects to strictly improve some agents' allocations.).
	
	Thus, for all agents $j\in N$, $B_{j} \succeq_{j} A_{j}$, and there exist some agent $i$ where $B_{i} \succ_{i} A_{i}$. 
	If for all $j\in N$, $|B_{j}| \geq |A_{j}|$ then $q'_{j} \geq q_{j}$. Now suppose for some $i$, $|B_{i}| > |A_{i}|$. This implies that $q'_{i} > q_{i}$. By adding these inequalities for all agents we have $\sum_{i} q'_{i} > \sum_{i} q_{i}$, contradicting the initial assumption of equal quota sizes ($C' = C$).
	
	For the rest of the proof, we consider two cases; one where the size of $B_{i}$ is greater than that of $A_{i}$, \ie, $|B_{i}| > |A_{i}|$, and one where $|B_{i}| \leq |A_{i}|$.
	
	\textbf{Case I}: Consider $|B_{i}| \leq |A_{i}|$ and $B_{i} \succ_{i} A_{i}$. If $B_{i} \subset A_{i}$ then monotonicity of lexicographic preferences in Lemma~\ref{Lem:lex} implies that $A_{i} \succ_{i} B_{i}$ contradicting the assumption. 
	On the other hand, if $B_{i} \not\subset A_{i}$ by Lemma~\ref{Lem:lex} there exists an object $x \in B_{i}$ such that for all $X \in  \mathbb{P}(B_{i} - A_{i})$ agent $i$ ranks it higher than any other subset, that is, $x \succ_{i} X$. In this case, serial dictatorship must also assign $x$ to agent $i$ in $A_{i}$, which is a contradiction.
	
	
	\textbf{Case II}: Consider $|B_{i}| > |A_{i}|$ and $B_{i} \succ_{i} A_{i}$. The proof of this case heavily relies on the lexicographic nature of preferences (as opposed to Case I that held valid for the class of monotonic, and not necessarily lexicographic, preferences).
	The inequality $|B_{i}| > |A_{i}|$ indicates that $q'_{i} > q_{i}$. We construct a preference profile $\succ'$ as follows: for each $j\in N$, if $B_{j} = A_{j}$ then $\succ'_{j} = \succ_{j}$, otherwise if $B_{j} \neq A_{j}$ rank the set $B_{j}$ higher than $A_{j}$ in $\succ'_{j}$ ($\succ'_{j} = B_{j} \succ A_{j} \succ \ldots$).
	Now run the serial dictatorship on $\succ'$ with quota $q$. Suppose that $B' = \pi(\succ')$. For agent $i$, $B'_{i}$ is the top $q_{i}$ objects of $B_{i}$ where $B'_{i} \subsetneq B_{i}$ and because $q_{i}$ is fixed, then $|B'_{i}| = |A_{i}|$. Given $\succ'$ we have $B_{i} \neq A_{i}$, which implies that $B'_{i} \neq A_{i}$. By strong monotonicity for agent $i$ we have $B_{i} \succ_{i} B'_{i} \succ_{i} A_{i}$. However, according to the constructed quotas we have $|B_{i}| > |B'_{i}|$ but $|B'_{i}| =  |A_{i}|$, where $B'_{i} \neq A_{i}$. By Lemma~\ref{Lem:lex} there exists an object $x \in B'_{i}$ which is preferred to all proper subsets of $A_{i} - B_{i}$. However, if such object exists it should have been picked by agent $i$ in the first place, which is in contradiction with agent $i$'s preference. \qed
\end{proof}

We state a few preliminary lemmas before proving our main result in characterizing the set of non-bossy, Pareto C-efficient, neutral, and strategyproof mechanisms. 
Given a non-bossy and strategyproof mechanism, an agent's allocation is only affected by her predecessor dictators. Thus, an agent's allocation may only change if the preferences of one (or more) agent with higher priority changes.

\begin{lemma}\label{lem:indep}
	Take any non-bossy and strategyproof mechanism $\pi$. Given two preference profiles $\succ, \succ' \in \mathcal{P}^{n}$ where $\succ = (\succ_{i}, \succ_{-i})$ and $\succ' = (\succ_{i}, \succ'_{-i})$, if for all $j < i$ we have $\pi_{f_{j}} (\succ) = \pi_{f_{j}}(\succ')$, then $\pi_{f_{i}} (\succ) = \pi_{f_{i}}(\succ')$.
\end{lemma}

The next Lemma guarantees that the outcome of a strategyproof and non-bossy mechanism only changes when an agent states that some set of objects that are less preferred to $\pi_{i}(\succ)$ under $\succ_{i}$ is now preferred under $\succ'_{i}$. 
Intuitively, any preference ordering $\succ'_{i}$ which reorders only the sets of objects that are preferred to $\pi_{i}(\succ)$ or the sets of objects that are less preferred to the set of objects allocated via $\pi_{i}(\succ)$ keeps the outcome unchanged. 

\begin{lemma}\label{lem:equal}
	Let $\pi$ be a strategyproof and non-bossy mechanism, and let $\succ,\succ' \in \mathcal{P}^{n}$. For all allocations $A\in \mathcal{A}$, if for all $i\in N, \pi_{i}(\succ) \succeq_{i} A_{i}$ and $\pi_{i}(\succ) \succeq'_{i} A_{i}$, then $\pi(\succ) = \pi(\succ')$.
\end{lemma}

The next lemma states that when all agents' preferences are identical, any strategyproof, non-bossy, and Pareto C-efficient mechanism simulates the outcome of a serial dictatorship quota mechanism.

\begin{lemma}\label{lem:ordering}
	Let $\pi$ be a strategyproof, non-bossy, and Pareto C-efficient mechanism with quota system $q$, and $\succ$ be a preference profile where all individual preferences coincide, that is $\succ_{i} = \succ_{j}$ for all $i,j \in N$. Then, there exists an ordering of agents, $f$, such that for each $k = 1,\ldots, |q|$, agent $f_{k}$ receives exactly $q_{k}$ items according to quota $q$ induced by a serial dictatorship.
\end{lemma}

\begin{theorem}\label{thm:serialDictatorships}
	If preferences are lexicographic, an allocation mechanism is strategyproof, non-bossy, neutral, and Pareto C-efficient if and only if it is a serial dictatorship quota mechanism.
\end{theorem}

\begin{proof}
	It is clear that in the multiple-assignment problem any serial dictatorship mechanism is strategyproof, neutral, and non-bossy~\cite{papai2001strategyproof}. For Pareto efficiency, in Proposition~\ref{prop:ParetoEfficient}, we showed that the serial dictatorship mechanism is Pareto C-efficient for any quota, and in fact it becomes Pareto efficient in a stronger sense when all objects are allocated $C = m$.
	
	Now, we must show that any strategyproof, Pareto C-efficient, neutral, and non-bossy mechanism, $\pi$, can be simulated via a serial dictatorship quota mechanism. Let $\pi$ be a strategyproof, Pareto C-efficient, neutral, and non-bossy mechanism. Consider $\succ\in\mathcal{P}^{n}$ to be an arbitrary lexicographic preference profile. Given $q$, we want to show that $\pi$ is a serial dictatorship mechanism. Thus, we need to find an ordering $f$ that induces the same outcome as $\pi$ when allocating objects serially according to quota $q$.

	Take an identical preference profile and apply the mechanism $\pi$ with a quota $q$.
	By Lemma~\ref{lem:ordering}, there exists a serial dictatorial allocation with an ordering $f$ where agent $f_1$ receives $q_{1}$ of her favorite objects from $M$, agent $f_2$ receives $q_{2}$ of her best objects from $M \setminus \pi_{f_{1}}$, and so on. Therefore, given a strategyproof, non-bossy, neutral, and Pareto C-efficient mechanism with quota $q$, we can identify an ordering of agents $f= (f_1, \ldots, f_{n})$ that receive objects according to $q = (q_1, \ldots, q_{n})$. Note that since the ordering is fixed a priori, the same $f$ applies to any non-identical  preference profile.
	
	From any arbitrary preference profile $\succ$, we construct an equivalent profile as follows: 
	Given the ordering $f$, the first best $q_{1}$ objects (the set of size $q_{1}$) according to $\succ_{f_{1}}$ are denoted by $A_{f_{1}}$ and are listed as the first objects (or set of objects of size $q_{1}$ since preferences are lexicographic) in $\succ'_{i}$. The next $q_{2}$ objects in $\succ'_{i}$ are the first best $q_{2}$ objects according to $\succ_{f_{2}}$ from $M \setminus A_{f_{1}}$, and so on. In general, for each $i= 2,\ldots, |q|$, the next best $q_{i}$ objects are the best $q_{i}$ objects according to $\succ_{f_{i}}$ from $M \setminus \bigcup_{j=1}^{j=i-1} A_{j}$. 
	Algorithm~\ref{alg:identical}, which can be found in the Appendix, illustrates these steps.
	
	Now we need to show that applying $\pi$ to the constructed identical preference profile ($\succ'$) induces the same outcome as applying it to $\succ$. 
	By Lemma~\ref{lem:indep} for each agent $f_{i}$, $\pi_{f_{i}}(\succ) = \pi_{f_{i}}(\succ')$ if for all $j < i$ we have $\pi_{f_{j}}(\succ) = \pi_{f_{j}}(\succ')$. That is, the allocation of an agent remains the same if the allocations of all previous agents remain unchanged.
	Now by Lemma~\ref{lem:equal}, for any allocation $A \in \mathcal{A}$, if for each agent $i\in N$, $\pi_{i}(\succ') \succeq'_{i} A_{i}$ then we also have $\pi_{i}(\succ') \succeq_{i} A_{i}$.
	For each $f_{i}$ where $i = 1,\ldots,|q|$, by Lemma~\ref{lem:equal} since $\pi$ is strategyproof and non-bossy, for any allocation $A_{f_{i}}$ given the quota $q$ we have $\pi_{f_{i}} \succeq'_{f_{i}} A_{f_{i}}$ and $\pi_{f_{i}} \succeq_{f_{i}} A_{f_{i}}$, which implies that $\pi_{f_{i}}(\succ') = \pi_{f_{i}}(\succ)$. Therefore, we have $\pi(\succ') = \pi(\succ)$.
	Since $\succ'$ is an identical profile, $\pi(\succ')= \pi(\succ)$ assigns $q_{i}$ objects to each agent according to the serial ordering $f$. 
	Thus, $\pi$ is a serial dictatorship quota mechanism. \qed
\end{proof}

The following example illustrates how an equivalent preference profile with identical outcome is constructed given any arbitrary preference profile, ordering, and quota system.

\begin{example}
	Consider allocating 4 objects to 3 agents with preferences illustrated in Table~\ref{Tab:identicalPref} (left), based on the following quota $q=(1,2,1)$. 
	Assume the following ordering of agents $f = (1,2,3)$. To construct a profile with identical orderings, agent 1's first best object according to $\succ_1$, $a$, is considered the highest ranking object in $\succ'_i$. Agent 2's best two objects ($q_{2} = 2$) among the remaining objects $c$ and $b$ are ranked next, and finally agent 3's remaining object $d$ is ranked last. Given $f$ and $q$, the two preference profiles depicted in Table~\ref{Tab:identicalPref} have exactly similar outcome (shown with squares).
	
	\begin{minipage}{\textwidth} 
		\centering
		\begin{minipage}{0.45\linewidth}
			\centering
			\begin{tabular}{|c c|}
				\hline 
				$\succ_{1}: $ & $ \boxed{a} \succ b \succ c \succ d$ \\ 
				
				$\succ_{2}: $ & $\boxed{c} \succ a \succ \boxed{b} \succ d$ \\ 
				
				$\succ_{3}: $ & $ a \succ c \succ \boxed{d} \succ b$ \\ 
				\hline 
			\end{tabular}
		\end{minipage}
		\begin{minipage}{0.45\linewidth}
			\centering
			\begin{tabular}{|c c|}
				\hline 
				$\succ'_{1}: $ & $ \boxed{a} \succ c \succ b \succ d$ \\ 
				
				$\succ'_{2}:  $ & $ a \succ \boxed{c} \succ \boxed{b} \succ d$ \\ 
				
				$\succ'_{3}:  $ & $ a \succ c \succ b \succ \boxed{d}$ \\ 
				\hline 
			\end{tabular}
		\end{minipage}
		\captionof{table}{Converting a preference profile to identical orderings, with exact same outcome.}
		\label{Tab:identicalPref}
	\end{minipage}
\end{example}

Finally, we show that strategyproofness and non-bossiness are necessary and sufficient conditions for group-strategyproofness.

\begin{proposition} \label{groupSP}
	A mechanism is group-strategyproof if and only if it is strategyproof and non-bossy.
\end{proposition}

It is critical to note that a group-strategyproof mechanism does not rule out the possibility of manipulation by a subset of agents that misreport their preferences and then exchange their allocations ex post. 
Consequently, it is easy to see that serial dictatorship quota mechanisms are guaranteed against group manipulation but do not prevent coalitional manipulation through reallocation \cite{papai2000hierarchical}.
We rewrite Theorem~\ref{thm:serialDictatorships} as the following:

\begin{theorem}\label{thm:groupSD}
	Serial dictatorship quota mechanisms are the only neutral, Pareto C-efficient, and group-strategyproof mechanisms.
\end{theorem}

\section{Sequential Dictatorship Quota Mechanisms}
In this section, we study a broader class of quota mechanisms by relaxing the neutrality requirement and allowing for the dictators to be identified in each sequence, as opposed to fixing the dictatorship orderings apriori.

\begin{proposition}\label{prop:seqPareto}
	A sequential dictatorship quota mechanism is Pareto C-efficient under lexicographic preferences.
\end{proposition}


The proof exactly follows as of the proof of Proposition~\ref{prop:ParetoEfficient}. 
%
%
Characterizing the set of strategyproof, non-bossy, and Pareto C-efficient quota mechanisms is similar to our characterization for serial dictatorship mechanisms, but requires a subtle change in Lemma~\ref{lem:ordering}.

\begin{lemma}\label{lem:SeqOrdering}
	Let $\pi$ be a strategyproof, non-bossy, and Pareto C-efficient mechanism with quota $q$, and $\succ$ be a preference profile where all individual preferences coincide, that is $\succ_{i} = \succ_{j}$ for all $i,j \in N$. 
	Then, there exists an ordering $f_{1}, f_{2}(\pi_{f_{1}}(\succ)), \ldots, f_{k}(\pi_{f_{1}}(\succ), \ldots, \pi_{f_{k-1}}(\succ))$ such that for each $i\in N$ agent $i$ receives exactly $q_{i}$ items according to quota $q$.
\end{lemma}


\begin{theorem}\label{thm:sequentialQD}
	An allocation mechanism is strategyproof, non-bossy, and Pareto C-efficient if and only if it is a sequential dictatorship quota mechanism.
\end{theorem}

	\begin{table}[t]
	\centering
	\begin{tabular}{|c c|}
		\hline 
		$\succ_{1}: $ & $ \boxed{a} \succ c \succ b$ \\ 
		
		$\succ_{2}: $ & $ \boxed{c} \succ b \succ a$ \\ 
		
		$\succ_{3}: $ & $ c \succ a \succ \boxed{b}$ \\ 
		\hline 
	\end{tabular}
	\caption{An example showing a mechanism that is group-strategyproof but not reallocation-proof.}
	\label{tab:reallocation}
\end{table}

\section{Randomized Quota Mechanisms}

So far we identified the class of deterministic strategyproof, non-bossy, and Pareto C-efficient quota mechanisms. 
However, deterministic quota mechanisms generally have poor fairness properties: the first dictator always has a strong advantage over the next dictator and so on. This unfairness could escalate when an agent gets to pick more objects than the successor agent, that is, $q_{i} > q_{j}$ for $i < j$.
Thus, while any profile-independent randomization over a set of serially dictatorial mechanisms still maintains the incentive property, randomization over priority orderings seem to be a proper way of restoring some measure of randomized fairness.

We first need to define a few additional properties in the randomized settings.
A random allocation is a stochastic matrix $A$ with $\sum_{i\in N} A_{i,j} = 1$ for each $j\in M$. This feasibility condition guarantees that the probability of assigning each object is a proper probability distribution. Moreover, every random allocation is a convex combination of deterministic allocations and is induced by a lottery over deterministic allocations~\cite{von1953certain}.
Hence, we can focus on mechanisms that guarantee Pareto C-efficient solutions ex post.

\begin{definition}[Ex Post C-Efficiency]
	A random allocation is ex post C-efficient if it can be represented as a probability distribution over deterministic Pareto C-efficient allocations.
\end{definition}

The support of any lottery representation of a strategyproof allocation mechanism must consist entirely of strategyproof deterministic mechanisms. Moreover, if the distribution over orderings does not depend on the submitted preferences of the agents, then such randomized mechanisms are strategyproof~\cite{roth1992two}.

We focus our attention on the \emph{downward} lexicographic dominance relation to compare the quality of two random allocations when preferences are lexicographic.\footnote{In the general domain, this measure corresponds to a stronger notion based on first-order stochastic dominance~\cite{bogomolnaia2001new,hadar1969rules}}
Given two allocations, an agent prefers the one in which there is a higher probability for getting the most-preferred object.
Formally, given a preference ordering $\succ_{i}$, agent $i$ prefers any allocation $A_{i}$ that assigns a higher probability to her top ranked object $A_{i,o_{1}}$ over any assignment $B_{i}$ with $B_{i,o_{1}} < A_{i,o_{1}}$, regardless of the assigned probabilities to all other objects.
Only when two assignments allocate the same probability to the top object will the agent consider the next-ranked object. Throughout this paper we focus on the downward lexicographic relation, as opposed to upward lexicographic relation \cite{Cho2016}. The downward lexicographic notion compares random allocations by comparing the probabilities assigned to objects in order of preference. Thus, it is a more natural way of comparing allocations and has shown to be often used in consumer markets and other settings involving human decision makers \cite{kahn1995exploratory,yee2007greedoid,tversky1988contingent}.


\begin{definition}
	Agent $i$ with preference $\succ_{i}$ downward lexicographically prefers random allocation $A_{i}$ to $B_{i}$ if 
	\begin{equation*}
		\exists\ \ell \in M: A_{i,\ell} > B_{i, \ell}\ \wedge\ \forall k \succ_{i} \ell:  A_{i,k} = B_{i, k}.
	\end{equation*} 
\end{definition}

We say that allocation $A$ \textbf{downward lexicographically dominates} another allocation $B$ if there exists no agent $i\in N$ that lexicographically prefers $B_{i}$ to $A_{i}$.
Thus, an allocation mechanism is downward lexicographically efficient (\emph{ld-efficient}) if for all preference profiles its induced allocation is not downward lexicographically dominated by any other random allocation. 
We can see that efficiency under general preferences immediately implies ld-efficiency under lexicographic preferences. However, some allocations may only guarantee efficiency when preferences are lexicographic.

\begin{example}\label{example:ld-domination}
	Consider four agents $N = \{1,2,3,4\}$ and four objects $M = \{a,b,c,d\}$ with quota $q=(1,1,1,1)$ at the following preference profile $\succ = ((cabd), (acdb), (cbda), (acbd))$. Note that preferences are only defined over single objects, and we write $(cabd)$ as a shorthand form of $\succ_{1} = c\succ a\succ b\succ d$.
	\begin{table}[t]
		\tabcolsep=0.1cm 
		\begin{subtable}{0.49\linewidth}
			\centering
			\begin{tabular}{ c c c c c}
				\hline
				& $a$ & $b$ & $c$ & $d$\\ \hline \hline
				$A_1$ & $0$ & $1/3$ & $1/2$ & $1/6$ \\
				$A_2$ & $1/2$ & $0$ & $0$ & $1/2$ \\
				$A_3$ & $0$ & $1/3$ & $1/2$ & $1/6$ \\
				$A_4$ & $1/2$ & $1/3$ & $0$ & $1/6$ \\
				\hline
			\end{tabular}
			\caption{sd-efficient allocation}
			\label{sdExample}
		\end{subtable}
		\begin{subtable}{0.49\linewidth}
			\centering
			\begin{tabular}{ c c c c c}
				\hline
				& $a$ & $b$ & $c$ & $d$\\ \hline \hline
				$A_1$ & $1/12$ & $1/3$  & $5/12$ & $1/6$  \\
				$A_2$ & $11/24$ & $0$  & $1/12$ & $11/24$  \\
				$A_3$ & $0$ & $5/12$ & $5/12$  & $1/6$ \\
				$A_4$ & $11/24$ & $1/4$ & $1/12$  & $5/24$ \\
				\hline
			\end{tabular}
			\caption{ld-dominated but not sd-dominated}
			\label{ldExample}
		\end{subtable} 
		\caption{An example showing an allocation that is ld-efficient but not sd-efficient.}
		\label{tab:ld}
	\end{table}
	Table~\ref{tab:ld} shows the stochastic efficient allocation in comparison with ld-efficient allocation.
	Here, even though the allocation in Table \ref{ldExample} is ld-dominated by the sd-efficient allocation, it is not stochastically dominated under the first-order stochastic dominance. This is because agent 2 (similarly agent 4) weakly prefers the allocation in Table \ref{ldExample} if only considering her first two top objects.
	Thus, the two random allocations are in fact incomparable with respect to stochastic dominance.
\end{example}

Given an allocation $A$, we say that agent $i$ is envious of agent $j$'s allocation if agent $i$ prefers $A_{j}$ to her own allocation $A_{i}$. Thus, an allocation is envyfree when no agent is envious of another agent's assignment. Formally we write,
\begin{definition}
	Allocation $A$ is {envyfree} if for all agents $i \in N$, there exists no agent-object pair $j\in N$, $\ell \in M$ such that,
	\begin{gather*}
		A_{j,\ell} > A_{i, \ell}\ \wedge\ \forall k \succ_{i} \ell:  A_{i,k} = A_{j,k}
	\end{gather*} 
\end{definition}
A mechanism is envyfree if at all preference profiles $\succ\in\mathcal{P}^{n}$ it induces an envyfree allocation.

\subsection{Random Serial Dictatorship Quota Mechanisms}

Recall that $|q|$ denotes the number of agents that are assigned at least one object.
Given a quota of size $|q|$, there are ${n \choose |q|}\times |q|!$ permutations (sequences without repetition) of $|q|$ agents from $N$. 
Thus, a \emph{Random Serial Dictatorship mechanism with quota $q$} is a uniform randomization over all permutations of size $|q|$. Formally, 

\begin{definition}[Random Serial Dictatorship Quota Mechanism (RSDQ)]
	Let $\mathbb{P}(N)$ be the power set of $N$, and $f\in \mathbb{P}(N)$ be any subset of $N$. Given a preference profile $\succ\in \mathcal{P}^{n}$, a random serial dictatorship with quota $q$ is a convex combination of serial dictatorship quota mechanisms and is defined as
	\begin{equation}
		\frac{\sum_{f\in \mathbb{P}(N):|f|= |q|} \pi_{f}(\succ)}{{n \choose |q|}\times |q|!}
	\end{equation}
\end{definition} 

In this randomized mechanism agents are allowed to pick more than one object according to $q$ and not all the agents may be allocated ex post. We can think of such mechanisms as extending the well-known Random Serial Dictatorship (RSD) for the house assignment problem wherein each agent is entitled to receive exactly one object. Thus, an RSD mechanism is a special case of our quota mechanism with $q_i = 1, \forall i \in N$ and $|q| = n$.

\begin{example}
	Consider three agents and four objects. Agents' preferences and the probabilistic allocation induced by RSDQ with quota $q=(2,1,1)$ are presented in Table \ref{exmpl:RDPexample}.
	Note that the size of $q$ can potentially be smaller than the number of agents, meaning that some agents may receive no objects ex post.
	
	\begin{minipage}{\textwidth}
		\centering
		
		\begin{minipage}{0.45\linewidth}
			\centering
			\begin{tabular}{|c|c|}
				\hline 
				$\succ_1$ & $c \succ a \succ b \succ d$ \\ 
				\hline 
				$\succ_2$ & $a \succ c \succ d \succ b$ \\ 
				\hline 
				$\succ_3$ & $c \succ b \succ d \succ a$ \\ 
				\hline 
			\end{tabular} 
		\end{minipage}
		\begin{minipage}{0.45\linewidth}
			\centering
			\begin{tabular}{ c c c c c}
				\hline
				& $a$ & $b$ & $c$ & $d$ \\ \hline \hline
				$A_1$ & $3/6$ & $1/6$ & $2/6$ & $2/6$  \\
				$A_2$ & $3/6$ & $0$ & $2/6$ & $3/6$  \\
				$A_3$ & $0$ & $5/6$ & $2/6$ & $1/6$  \\
				\hline
			\end{tabular}
		\end{minipage}
		\captionof{table}{RSDQ allocation with $q = (2,1,1)$.}  \label{exmpl:RDPexample}
	\end{minipage}
	

\end{example}

The weakest notion of fairness in randomized settings is the equal treatment of equals. We say an allocation is fair (in terms of equal treatment of equals) if it assigns an identical random allocation (lottery) to agents with identical preferences.

\begin{theorem}\label{thm:anySD}
	Take any serial dictatorship mechanism $\pi$ with a quota $q$. A uniform randomization over all permutations of orderings with size $|q|$ is strategyproof, ex post C-efficient, and fair (equal treatment of equals).
	
\end{theorem}


Now, we present our main result for envyfreeness of RSDQ regardless of the selected quota system.

\begin{theorem}\label{thm:RSDQ-Envy}
	Random Serial Dictatorship Quota mechanism is envyfree with any quota $q$, under downward lexicographic preferences.
\end{theorem}

\begin{proof}
	Let $A$ denote a random allocation induced by RSDQ with quota $q$ at an arbitrary preference profile $\succ \in \mathcal{P}^{n}$. 
	Suppose for contradiction that there exists an agent $i\in N$ with random allocation $A_{i}$ that prefers another agent's random allocation $A_{j}$ to her own assignment, that is, $A_{j} \succ_{i} A_{i}$.
	%
	Assuming that preferences are downward lexicographic, there exists an object $\ell$ such that $A_{j,\ell} > A_{i, \ell}$ and for all objects that are ranked higher than $\ell$ (if any) they both receive the same probability $\forall k \succ_{i} \ell:  A_{i,k} = A_{j,k}$. Thus, we can write:
	$
	\sum_{x\in A_{i}: x\succ_{i} \ell} A_{j,x} = \sum_{x\in A_{i}: x \succ_{i} \ell} A_{i,x}
	$. %
	Since preferences are lexicographic, the assignments of objects less preferred to $\ell$ become irrelevant because for two allocations $A_{i}$ and $B_{i}$ such that $A_{i,\ell} > B_{i,\ell}$, we have $A_{i} \succ_{i} B_{i}$ for all $x\prec_{i} \ell$ where $B_{i,x} \geq A_{i,x}$. Thus, we need only focus on object $\ell$.
	
	Let $\mathcal{F}$ denote the set of all orderings of agents where $i$ is ordered before $j$ or $i$ appears but not $j$. Note that since we allow for $|q| = |f| \leq n$, some agents could be left unassigned, and permuting $i$ and $j$ could imply that one is not chosen under ${n \choose |q|}$.
	For any ordering $f\in \mathcal{F}$ of agents where $i$ precedes $j$, let $\bar{f} \in \bar{\mathcal{F}}$ be the ordering obtained from $f$ by swapping $i$ and $j$. Clearly, $|\mathcal{F}| = |\bar{\mathcal{F}}|$ and the union of the two sets constitute the set of orderings that at least one of $i$ or $j$ (or both) is present. Fixing the preferences, we can only focus on $f$ and $\bar{f}$.
	
	Let $\pi_{f}(\succ)$ be the serial dictatorship with quota $q$ and ordering $f$ at $\succ$. 
	RSDQ is a convex combination of such deterministic allocations with equal probability of choosing an ordering from any of $\mathcal{F}$ or $\bar{\mathcal{F}}$. 
	
	Given any object $y\in M$, either $i$ receives $y$ in $\pi_{f}$ and $j$ gets $y$ in $\pi_{\bar{f}}$, or none of the two gets $y$ in any of $\pi_{f}$ and $\pi_{\bar{f}}$.
	Thus, object $\ell$ is either assigned to $i$ in $\pi_{f}$ and to $j$ in $\pi_{\bar{f}}$, or is assigned to another agent.
	If $i$ gets $\ell$ in $\pi_{f}$ for all $f\in \mathcal{F}$, then $j$ receives $\ell$ in $\pi_{\bar{f}}$. The contradiction assumption $A_{j,\ell} > A_{i,\ell}$ implies that there exists an ordering $f$ where $i$ receives a set of size $q_{i}$ that does not include object $\ell$ while $j$'s allocation set includes $\ell$. Let $X_{i}$ denote this set for agent $i$ and $X_{j}$ for agent $j$. Then, $X_{i} \succ_{i} X_{j}$. Thus, by definition there exists an object $\ell' \in X_{i}$ such that $\ell' \succ_{i} \ell$, where $\ell' \not\in X_{j}$. Thus, the probability of assigning object $\ell' \succ_{i} \ell$ 
	to $i$ is strictly greater than assigning it to $j$, that is, $A_{i,\ell'} > A_{j, \ell'}$. However, by lexicographic assumption we must have $\forall k \succ_{i} \ell:  A_{i,k} = A_{j,k}$, which is a contradiction. \qed
\end{proof}

\begin{theorem}\label{thm:RSDQ-Charac}
	Under downward lexicographic preferences, a Random Serial Dictatorship Quota mechanism is ex post C-efficient, strategyproof, and envyfree for any number of agents and objects and any quota system.
\end{theorem}

The well-known random serial dictatorship mechanism (RSD), also known as Random Priority, is defined when $n=m$ and assigns a single object to agents~\cite{abdulkadirouglu1998random}. It is apparent that RSD is a special instance from the class of RSDQ mechanisms. 

\begin{corollary} \label{cor:RSD}
	RSD is ex post efficient, strategyproof, and envyfree when preferences are downward lexicographic.
\end{corollary}


\section{Discussion}

We investigated strategyproof allocation mechanisms when agents with lexicographic preferences may receive more than one object according to a quota.
The class of sequential quota mechanisms enables the social planner to choose any quota without any limitations. 
For the general domain of preferences, however, the class of strategyproof, non-bossy, and Pareto efficient mechanisms is restricted to sequential dictatorships with equal quota sizes. Demanding neutrality, the set of such mechanisms gets restricted to quasi-dictatorial mechanisms, which are far more unfair~\cite{papai2000strategyproof,papai2001strategyproof}. Thus, such mechanisms limit a social planner to specific quota systems while demanding the complete allocation of all available objects.

We showed that the class of strategyproof allocation mechanisms that satisfy neutrality, Pareto C-efficiency, and non-bossiness expands significantly when preferences are lexicographic. Our characterization shows that serial dictatorship quota mechanisms are the only mechanisms satisfying these properties in the multiple-assignment problem. Removing the neutrality requirement, this class of mechanisms further expands to sequential dictatorship quota mechanisms.

To recover some level of fairness, we extended the serial dictatorship quota mechanisms to randomized settings and showed that randomization can help achieve some level of stochastic symmetry amongst the agents. More importantly, we showed that RSDQ mechanisms satisfy strategyproofness, ex post C-efficiency, and envyfreeness for any number of agents, objects, and quota systems when preferences are downward lexicographic. 
The envyfreeness result is noteworthy: it shows that in contrast to the Probabilistic Serial rule (PS)~\cite{bogomolnaia2001new} which satisfies strategyproofness when preferences are lexicographic only when $n \geq m$~\cite{schulman2012allocation}, the well-known RSD mechanism in the standard assignment problem is envyfree for any combination of $n$ and $m$. These results address the two open questions about the existence of a mechanism with more favorable fairness and strategyproofness properties~\cite{papai2000strategyproof,schulman2012allocation}.

%
%
%

Serial dictatorship mechanisms are widely used in practice since they are easy to implement while providing stability and strategyproofness guarantees~\cite{roth1995handbook}. Serial dictatorship quota mechanisms and their randomized counterparts provide a richer framework for multiple allocation problems while creating the possibility of fair and envyfree assignments.
Our characterization for deterministic quota mechanisms when preferences are lexicographic justifies the use of quotas in sequential settings. 
In randomized settings, however, an open question is whether RSDQ mechanisms are the only allocation rules that satisfy the above properties in the multiple assignment domain. Of course, answering this question, first, requires addressing the open question by Bade~\cite{bade2014random} in the standard assignment problem (where every agent gets at most one object): \textit{is random serial dictatorship a unique mechanism that satisfies strategyproofness, ex post efficiency, and equal treatment of equals?}

%
%
%

 
\bibliography{references}
\bibliographystyle{splncs}

%

\clearpage
\appendix
\spnewtheorem*{defn}{Definition}{\bfseries}{\rmfamily}
\spnewtheorem*{prop}{Proposition}{\bfseries}{\rmfamily}
\spnewtheorem*{thm}{Theorem}{\bfseries}{\rmfamily}
\spnewtheorem*{lem}{Lemma}{\bfseries}{\rmfamily}

\section{Missing Proofs and Algorithms}
In this appendix we include all proofs that were missing in the main part of the paper, along with the pseudocode used to construct particular preference profiles that are needed in some of the proofs.

\subsection{Proof of Theorem \ref{thm:noInterleaving}}

\begin{proof}
	The proof follows by constructing a manipulable preference profile. Given any Pareto C-efficient and non-bossy interleaving mechanism, we show that we can construct an instance (preference profile) at which at least one agent can manipulate the outcome. 
	
	Suppose there is a non-bossy and Pareto C-efficient mechanism $\pi$ with at least one alternation between agents $i$ and $j$. Note that the alternation could be through a fixed ordering or through a picking process. Since we are constructing an instance, we can assume that all other agents $k\in N\setminus \{i,j\}$ will receive their objects after agents $i$ and $j$ (or have already received their non-conflicting objects before the two). We now  construct a preference profile such that $\succ = (\succ_{i}, \succ_{j}, \succ_{N\setminus \{i,j\} })$. 
	
	Let $f_{k}$ denote the agent in the $k$th picking order, that is, $f_{2} = i$ indicates that the agent in the second picking order is agent $i$.
	Consider the ordering such that for agents 1 and 2 we have $f_1 = f_3 = 1$ and $f_2 = 2$. Assume there are 3 objects available and construct a preference profile as follows: $\succ_1 = a \succ b \succ c$ and $\succ_{2} = o_{1} \succ o_{2} \succ o_{3}$, where $o_{k}$ represent the $k$th ranked object in $\succ_2$.
	By Pareto C-efficiency and non-bossiness of $\pi$, agents final allocations must preclude any further exchange between the two agents, and no agent can change the allocation of the other while its own allocation remains unchanged.
	
	Since agent 1 picks first and last according to ordering $f$, agent 1 can pick her first choice either at stage 1 or 3 as long as agent 2's top choice is not equal to that of agent 1, \ie $o_{1} \in \{b,c\}$. If $o_{1} = c$ then there is no conflict between agent 1 and 2 and playing truthfully has the best outcome for agent 1. Thus, it follows that $o_{1} = b$ and $o_{2} \in \{a,c\}$. Now we need to construct the rest of agent 2's ordering such that agent 1's top choice, object $a$, remains in the pool of objects until the last stage. Thus, for the following profile $\succ_1= a \succ b \succ c$ and $\succ_2 = b \succ c \succ a$, the interleaving mechanism is manipulable. This implies that no Pareto C-efficient and non-bossy interleaving mechanism guarantees strategyproofness. \qed
\end{proof}

\subsection{Proof of Lemma \ref{Lem:lex}}

\begin{proof}
	We provide proof for each of the statements in the lexicographic domain.
	\begin{itemize}
		\item Since $B_{i} \subset A_{i}$ then all objects in $B_{i}$ are also in $A_{i}$, and there exists an object $x \in A_{i}$ such that $x \notin B_{i}$. By the definition of lexicographic preferences, having an object is preferred to not having the object (\ie objects are goods). Therefore, $A_{i} \succ_{i} B_{i}$.
		
		\item It is easy to see that adding a set of object $X \cap A_{i} = \emptyset$ to two sets such that $A_{i} \succ_{i} B_{i}$ maintains the preference over the two sets.
		This is because elements in $X$ are added to both sets and by assumption there is still an element $x\in A_{i}$ and $x\notin X$ that is preferred to all objects in $B_{i}$.
		We should prove the converse that if $A_{i} \cup X \succ_{i} B_{i} \cup X$ then $A_{i} \succ_{i} B_{i}$. 
		Suppose not, that is $B_{i} \succeq_{i} A_{i}$. By adding $X = B_{i} - A_{i}$ to both sides we have $B_{i}\cup X \succeq_{i} A_{i} \cup X$, that is, $B_{i} \succeq_{i} A_{i} \cup B_{i}$, which contradicts the strong monotonicity of lexicographic preferences when $A_{i}$ is nonempty.
		
		\item Suppose that there does not exist an object $x\in A_{i}$ such that $x \succ_{i} X$ for all $X \in  \mathbb{P}(B_{i} - A_{i})$. The set $X$ can be any power set of $B_{i} - A_{i}$, and for the sake of this proof we assume that $X = B_{i} - A_{i}$.
		By the second statement in this lemma, for $A_{i} \succ_{i} B_{i}$, we can add any $X$ such that $X\cap A_{i} = \emptyset$ to the both sides and write $A_{i} \cup X \succ_{i} B_{i} \cup X$, which holds since $X = B_{i} - A_{i}$.
		This states that for any object $x \in B_{i}$, $x$ is also a member of $A_{i} \cup X$, implying that $B_{i} \subset A_{i} \cup X$. Note that $B_{i} \neq A_{i} \cup X$ because  $A_{i}$ is considered to be nonempty.
		Using the first statement in this lemma, if $B_{i} \subset A_{i} \cup X$ then $A_{i} \cup X \succ_{i} B_{i}$. Replacing $X$ with $B_{i} - A_{i}$ and subtracting it from both sides, we have $A_{i} \succ_{i} \emptyset$, which implies that there exists an object $x\in A_{i}$ such that $x \notin B_{i}$ and $x \succ_{i} B_{i} - A_{i}$, contradicting the initial assumption.
		%
		%
	\end{itemize} 
	The above items conclude our proof for the statements in this lemma. \qed
\end{proof}

\subsection{Proof of Lemma \ref{lem:indep}}

\begin{proof}
	For all $j < i$ we have $\pi_{f_{j}} (\succ) = \pi_{f_{j}}(\succ')$. By non-bossiness and strategyproofness, for all $\succ'_{j}$ such that $\pi_{j}(\succ) = \pi_{j}(\succ'_{j},\succ_{-j})$ we have $\pi(\succ) = \pi(\succ'_{j},\succ_{-j})$. 
	In words, non-bossiness and strategyproofness prevent any agent to change the allocation of other agents with lower priority (those who are ordered after him), without changing its own allocation.
	Let $M'$ be the set of remaining objects such that $M' =  M \setminus \bigcup_{k=1}^{j} \pi_{f_{k}}(\succ)$. Since $\pi_{f_{j}} (\succ) = \pi_{f_{j}}(\succ')$, the set of remaining objects $M'$ under $\succ'$ is equivalent to those under $\succ$, implying that $\pi_{f_{i}} (\succ) = \pi_{f_{i}}(\succ')$ which concludes the proof. \qed
\end{proof}

\subsection{Proof of Lemma~\ref{lem:equal}}

\begin{proof}
	
	The proof follows similar to Lemma 1 in \cite{svensson1999strategy}. 
	First, we show that $\pi(\succ'_{i}, \succ_{-i}) =  \pi(\succ)$, that is changing $i$'s preference only does not affect the outcome. From strategyproofness we know that 
	$
	\pi_{i}(\succ_{i}) \succeq_{i} \pi_{i}(\succ'_{i}, \succ_{-i})
	$. 
	By the lemma's assumption (if condition) we can also write 
	$
	\pi_{i}(\succ_{i}) \succeq'_{i} \pi_{i}(\succ'_{i}, \succ_{-i})
	$. 
	However, strategyproofness implies that 
	$
	\pi_{i}(\succ'_{i}, \succ_{-i}) \succeq'_{i} \pi_{i}(\succ_{i})
	$. 
	Since the preferences are strict, the only way for the above inequalities to hold is when $\pi_{i}(\succ'_{i}, \succ_{-i}) =  \pi_{i}(\succ)$. The non-bossiness of $\pi$ implies that $\pi(\succ'_{i}, \succ_{-i}) =  \pi(\succ)$.
	
	We need to show that the following argument holds for all agents. We do this by partitioning the preference profile into arbitrary partitions constructed partly from $\succ$ and partly from $\succ'$. Let $\succ^{p} = (\succ'_{1}, \ldots, \succ'_{p-1}, \succ_{p},\ldots, \succ_{n}) \in \mathcal{P}^{n}$. Thus, a sequence of preference profiles can be recursively written as $\succ^{p+1} = (\succ'_{p}, \succ^{p}_{-p})$.
	Using the first part of the proof and by the recursive representation, we can write $\pi(\succ^{p}) = \pi(\succ'_{p}, \succ^{p}_{-p}) = \pi(\succ^{p+1})$. Now using this representation, we shall write $\pi(\succ') = \pi(\succ^{n+1})$ and $\pi(\succ) = \pi(\succ^{1})$, which implies that $\pi(\succ) = \pi(\succ')$. \qed
\end{proof}

\subsection{Proof of Lemma \ref{lem:ordering}}

\begin{proof}
	Suppose the contrary and let $\succ$ be an identical preference profile $\succ_{1} = \succ_{2} = a \succ b \succ c$ such that agent 1 receives $a$ and $c$ while agent 2 receives $b$.
	For agents 1 and 2, assume that they both have received no other objects except the ones stated above (Alternatively, we can assume that the other objects received by these two agents so far are their highest ranked objects, and because these objects were assigned in some previous steps, they won't affect the assignment of the remaining objects). For all other agents $N\setminus \{1,2\}$ assume that the allocation remains unchanged, \ie, these agents will receive exactly the same objects after we change the preferences of agent 1. 
	By Lemma~\ref{lem:equal}, since the mechanism is non-bossy and strategyproof, agent 1's allocation remains unchanged under the following changes in its preference ordering:
	\begin{gather*}
	\succ_{1} = a \succ b \succ c \Rightarrow a \succ c \succ b \Rightarrow c \succ a \succ b
	\end{gather*}
	Thus, the new preference profile $\succ'$ would be
	\begin{equation*}
	\begin{array}{cc}
	\succ'_{1}:  &  \boxed{c} \succ \boxed{a} \succ b  \\ 
	\succ_{2}:  & a \succ \boxed{b} \succ c
	\end{array}
	\end{equation*}
	where $\pi(\succ') = \pi(\succ)$. The squares show the current allocation. Since agent 1 is receiving two objects and agent 2 receives one, for any ordering that is not prescribed by a serial dictatorship, agent 2 should be ordered second (otherwise, the ordering is a serial dictatorship).
	
	More specifically, orderings (1,2) and (2,1) are serial dictatorships. Since agent 2 must be ordered second, it must be the case that agent 1 goes first and third (otherwise we are back at (1,2), which results in a serial dictatorship).
	Agent 1 first chooses object $c$ according to $\succ'_{1}$, then agent 2 chooses object $a$ according to $\succ_{2}$, and lastly agent 1 chooses the remaining object $b$.
	Therefore, agent 2 can benefit from manipulating the mechanism by choosing $a$ instead of $b$, contradicting the assumption that $\pi$ is strategyproof and non-bossy. This implies that such agents cannot exist, and concludes our proof. \qed
\end{proof}

\subsection{Algorithm~\ref{alg:identical}}

	\begin{algorithm}[H]
		\caption{Constructing an identical preference profile}
		\label{alg:identical}
		\KwData{A preference profile $\succ$, an ordering $f$, and quota $q$}
		\KwResult{A profile with identical preferences $\succ'$ with $\pi(\succ') = \pi(\succ)$}
		Initialize $\succ_{1} \leftarrow \emptyset$ \\
		Initialize set $Z = \emptyset$\\
		\For{$(i \leftarrow 1$ \KwTo $|q|)$}{ 
			$Z \leftarrow \text{top}(q_{i}, \succ_{f_{i}})$  \emph{{\footnotesize // Most preferred set of size $q_{i}$ from the remaining objects.}} \\
			$\succ'_{1} \leftarrow \text{append}(\succ'_{1}, Z)$ \emph{{\footnotesize // Append this set to the preference ordering.}} \\
			$Z \leftarrow \emptyset$\\
		}
		\For{$(i \leftarrow 1$ \KwTo $|f|)$}{
			$\succ'_{i} \leftarrow \succ'_{1}$\\
		}
		\textbf{return} $\succ'$.
	\end{algorithm}

\subsection{Proof of Proposition \ref{groupSP}}

\begin{proof}
	It is easy to see that group-strategyproofness implies strategyproofness and non-bossiness. We need to show the converse, that is, if $\pi$ is strategyproof and non-bossy then it is group-strategyproof.\footnote{The proof is inspired by Lemma 1 in \cite{papai2000hierarchical} for single-object allocation and extends the agent allocations to sets of objects.}
	Let $N' \subseteq N$ be a subset of agents, $N' = \{1, \ldots, n'\}$, with $\succ'_{N'}$ such that allocation of some agents in $N'$ strictly improves while for other agents in $N'$ the allocation remains the same. Formally, for all $i\in N'$, $\pi_{i}(\succ'_{N'}, \succ_{-N'}) \succeq_{i} \pi_{i}(\succ)$ and for some $j\in N'$,  $\pi_{j}(\succ'_{N'}, \succ_{-N'}) \succ_{j} \pi_{j}(\succ)$.
	Construct an alternative preference profile $\hat{\succ}$ such that for all $i\in N'$ the preference ordering $\hat{\succ}_{i}$ preserves the ordering but moves the set $\pi_{i}(\succ'_{N'}, \succ_{-N'})$ to the first rank in the ordering.
	
	For agent $1$, if $\pi_{1}(\succ'_{N'}, \succ_{-N'}) \succ_{1} \pi_{1}(\succ)$ then by Lemma~\ref{lem:indep}, $\pi_{1}(\succ'_{N'}, \succ_{-N'})$ is not in the list of available sets. Otherwise, $\pi_{1}(\succ'_{N'}, \succ_{-N'}) = \pi_{1}(\succ)$. Thus, strategyproofness implies that $\pi_{1}(\hat{\succ}_{1}, \succ_{-1}) = \pi_{1}(\succ)$, and by non-bossiness we have $\pi(\hat{\succ}_{1}, \succ_{-1}) = \pi(\succ)$.
	Repeating the same argument for all other agents in $\{2, \ldots, n'\}$, we get $\pi(\hat{\succ}_{N'}, \succ_{-N'}) = \pi(\succ)$. 
	Now since $\pi$ is strategyproof and non-bossy, using Lemma~\ref{lem:equal} we have that $\pi(\hat{\succ}_{N'}, \succ_{-N'}) = \pi(\succ'_{N'}, \succ_{-N'})$. This implies that $\pi(\succ_{N'}, \succ_{-N'}) = \pi(\succ)$, meaning that $\pi$ is group-strategyproof. \qed
\end{proof}

\subsection{Proof of Lemma \ref{lem:SeqOrdering}}

\begin{proof}
	Let $\pi$ be a strategyproof, non-bossy, and Pareto C-efficient mechanism with quota $q$. By Lemma~\ref{lem:ordering}, we know that for each identical preference profile, there exists a fixed ordering $f': (f'_{1}, \ldots, f'_{k})$ such that agent $f'_{1}$ receives $q_{1}$ objects, agent $f'_{2}$ receives $q_{2}$, and so on. 
	Let $f$ be a dictatorship ordering such that $f_{1}, f_{2}(\pi_{f_{1}}(\succ)), \ldots, f_{k}(\pi_{f_{1}}(\succ), \ldots, \pi_{f_{k-1}}(\succ))$. We show that for each ordering of agents, there is an exact mapping from $f'$ to $f$. For all preference profiles, map each agent ordering as follows: $f_{1} =  f'_{1}$, $f_{2}(\pi_{f_{1}}(\succ)) = f'_{2}$, $\ldots$, $f_{k}(\pi_{f_{1}}(\succ), \ldots, \pi_{f_{k-1}}(\succ) = f'_{k}$. This implies that $f$ is a dictatorial ordering, which concludes our existence proof. \qed
\end{proof}

\subsection{Proof of Theorem~\ref{thm:sequentialQD}}

\begin{proof}
	Sequential dictatorship quota mechanisms are strategyproof and non-bossy. Proposition~\ref{prop:seqPareto} states that when preferences are lexicographic sequential dictatorships are Pareto C-efficient. Sequential dictatorships are also Pareto efficient when $C = \sum_{i=1}^{|q|} q_{i}$.
	
	We must show the converse. Let $\pi$ be a strategyproof, Pareto C-efficient, and non-bossy mechanism with quota $q$. By Lemma~\ref{lem:SeqOrdering}, given an identical preference profile and a quota $q$, there exists a sequential ordering $f$ where agent $f_1$ receives $q_{1}$ of her favorite objects from $M$, agent $f_{2}(\pi_{f_{1}}(\succ))$ receives $q_{2}$ of her best objects from $M \setminus\pi_{f_{1}}$, and so on. Therefore, since the choice of the first dictator is independent of preference profile, we can identify a sequential ordering $f_{1}, f_{2}(\pi_{f_{1}}(\succ)), \ldots, f_{k}(\pi_{f_{1}}(\succ), \ldots, \pi_{f_{k-1}}(\succ))$ that receive objects according to $q = (q_1, \ldots, q_{k})$.
	
	Similar to the proof of Theorem \ref{thm:serialDictatorships}, we construct an alternate preference profile $\succ'$, based on the given preference profile, at which all agents have identical preferences, where $\succ' = (\succ'_{i}, \ldots,\succ'_{i})$. 
	
	According to function $f$, the first best $q_{1}$ objects according to $\succ_{f_{1}}$ are denoted by $\pi_{f_{1}}(\succ)$ and are listed as the first objects in $\succ'_{i}$. The next $q_{2}$ objects in $\succ'_{2}$ are the first best $q_{2}$ objects according to $\succ_{f_{2}(\pi_{f_{1}}(\succ))}$ from $M \setminus \pi_{f_{1}}(\succ)$, and so on. In general, for each $i\in N\setminus f_{1}$, the next best $q_{i}$ objects are the best $q_{i}$ objects according to $\succ_{f_{i}(\pi_{f_{1}}(\succ), \ldots, \pi_{f_{i-1}}(\succ))}$ from $M \setminus \bigcup_{j=1}^{j=i-1} \pi_{f_{j}}(\succ)$. These steps are depicted in Algorithm \ref{alg:Seqidentical}.
	
	By Lemma~\ref{lem:indep}, for any agent in $f$ the outcome of $\pi(\succ')$ must remain unchanged if the outcome of all predecessor agents remains unchanged. Thus, by Lemma~\ref{lem:equal}, for any allocation $A \in \mathcal{A}$, if for each agent $i\in N$, $\pi_{i}(\succ') \succeq'_{i} A_{i}$ then we also have $\pi_{i}(\succ') \succeq_{i} A_{i}$. 
	For each $f_{i}(\cdot)$ where $i = 1,\ldots,|f|$, by Lemma~\ref{lem:equal} since $\pi$ is strategyproof and non-bossy, for any allocation $A_{f_{i}}$ given the quota $q$ we have 
	\begin{gather*}
	\pi_{f_{i}(\pi_{f_{1}}(\succ), \ldots, \pi_{f_{i-1}}(\succ))} \succeq'_{f_{i}(\pi_{f_{1}}(\succ), \ldots, \pi_{f_{i-1}}(\succ))} A_{f_{i}(\pi_{f_{1}}(\succ), \ldots, \pi_{f_{i-1}}(\succ))} \\
	\pi_{f_{i}(\pi_{f_{1}}(\succ), \ldots, \pi_{f_{i-1}}(\succ))} \succeq_{f_{i}(\pi_{f_{1}}(\succ), \ldots, \pi_{f_{i-1}}(\succ))} A_{f_{i}(\pi_{f_{1}}(\succ), \ldots, \pi_{f_{i-1}}(\succ))}
	\end{gather*}
	which implies that $\pi(\succ') = \pi(\succ)$. Therefore, we identified an sequential ordering of agents that induces the same outcome as the original mechanism. Thus, $\pi$ is a sequential dictatorship quota mechanism. \qed
\end{proof}

\begin{algorithm}
	\caption{Constructing an identical preference profile}
	\label{alg:Seqidentical}
	\KwData{A preference profile $\succ$, first dictator $f_{1}$, and quota $q$}
	\KwResult{A profile with identical preferences $\succ'$ with $\pi(\succ') = \pi(\succ)$}
	Initialize $\succ_{1} \leftarrow \emptyset$ \\
	Initialize set $Z = \emptyset$\\
	\For{$(i \leftarrow 1$ \KwTo $|q|)$}{ 
		\If{$(i = 1)$}{
			$k \leftarrow f_{1}$ \emph{{\footnotesize // The first dictator is known.}} \\
		}
		\Else{
			$k \leftarrow f_{i}(\pi_{f_{1}}(\succ), \ldots, \pi_{f_{i-1}}(\succ))$ \emph{{\footnotesize // Identify the next dictator}} \\
		}
		$Z \leftarrow \text{top}(q_{i}, \succ_{k})$  \emph{{\footnotesize // Most preferred set of size $q_{i}$ from the remaining objects.}} \\
		$\succ'_{1} \leftarrow \text{append}(\succ'_{1}, Z)$ \emph{{\footnotesize // Append this set to the preference ordering.}} \\
		$Z \leftarrow \emptyset$\\
	}
	\For{$(i \leftarrow 1$ \KwTo $|f|)$}{
		$\succ'_{i} \leftarrow \succ'_{1}$\\
	}
	\textbf{return} $\succ'$.
\end{algorithm}

\subsection{Proof of Theorem \ref{thm:anySD}}

\begin{proof}
	Showing ex post C-efficiency is simple: any serial dictatorship mechanism satisfies Pareto C-efficiency, and thus, any randomization also guarantees a Pareto C-efficient solution ex post.
	The support of the random allocation consists of only strategyproof deterministic allocations, implying that the randomization is also strategyproof. The equal treatment of equal is the direct consequence of the uniform randomization over the set of possible priority orderings. \qed
\end{proof}

\subsection{Proof of Corollary \ref{cor:RSD}}

\begin{proof}
	The conventional RSD mechanism is equivalent to an RSDQ mechanism where agents receive exactly one object, that is, $\sum_{i} q_{i} = m$ and for each agent $i$, $q_{i} = 1$. Therefore, RSD satisfies ex post efficiency, strategyproofness, and envyfreeness. \qed
\end{proof}

\end{document}